\documentclass[a4paper,numberwithinsect,USenglish]{lipics}

\usepackage{amsmath}
\usepackage{amssymb}
\usepackage{xspace}
\usepackage{stmaryrd}
\usepackage{paralist}
\usepackage{thm-restate}
\usepackage{hyperref}
\usepackage{graphicx}

\usepackage{microtype}

\newcommand{\bisim}[1]{\ensuremath{\rightleftharpoons}_{#1}}
\newcommand{\md}[1]{\ensuremath{\mathop{md}}(#1)}
\newcommand{\set}[1]{\ensuremath\left\{#1\right\}}
\newcommand{\ddfn}{\mathrel{\mathop{{\mathop:}{\mathop:}}}=}
\newcommand{\dep}[1]{\mathord{=}\!\left(#1\right)}

\newcommand{\FO}{\ensuremath{\mathsf{FO}}\xspace}
\newcommand{\D}{\ensuremath{\mathsf{D}}\xspace}
\newcommand{\MTL}{\ensuremath{\mathsf{MTL}}\xspace}
\newcommand{\MLFO}{\ensuremath{\mathsf{ML}^{\mathsf{FO}}}\xspace}
\newcommand{\EMINCL}{\ensuremath{\mathsf{EMINCL}}\xspace}
\newcommand{\EMINCLC}{\ensuremath{\mathsf{EMINCL}(\ovee)}\xspace}
\newcommand{\MDL}{\ensuremath{\mathsf{MDL}}\xspace}
\newcommand{\ML}{\ensuremath{\mathsf{ML}}\xspace}

\newcommand{\MIL}{\ensuremath{\mathsf{MIL}}\xspace}
\newcommand{\EMDL}{\ensuremath{\mathsf{EMDL}}\xspace}
\newcommand{\EMIL}{\ensuremath{\mathsf{EMIL}}\xspace}
\newcommand{\EMILC}{\ensuremath{\mathsf{EMIL}(\ovee)}\xspace}
\newcommand{\card}[1]{\left| #1 \right|}
\newcommand{\mathtext}[1]{\ensuremath{\mathrm{\text{#1}}}}
\newcommand{\nmodels}{\ensuremath{\not\models}}
\newcommand{\existsWorld}{\ensuremath{\mathsf E}}

\newcommand{\propVar}{\ensuremath{X}}
\newcommand{\foStructure}[2]{\ensuremath{\mathcal M^{\mathtext{FO}}_{#1,#2}}}
\newcommand{\mDiss}{\ensuremath{M^{\mathtext{DISS}}}}
\newcommand{\hintikkaFormulasLevel}[1]{\ensuremath{\Phi^{\bisim #1}}}
\newcommand{\hintikkaPointedModel}[3]{\ensuremath{\phi^{#3}_{#1,#2}}}
\newcommand{\hintikkaSetOfPointedModelsInTeam}[3]{\ensuremath{\Phi^{\bisim #3}_{#1,#2}}}
\newcommand{\hintikkaModelWithTeam}[3]{\ensuremath{\varphi^{\bisim #3}_{#1,#2}}}

\newcommand{\classNegation}{\ensuremath{\mathord{\sim}}}

\theoremstyle{plain}
\newtheorem{proposition}[theorem]{Proposition}

\title{A Van Benthem Theorem for Modal Team Semantics\footnote{The first author was supported by the Academy of Finland grants 264917 and 275241.}}

\author[1]{Juha Kontinen}
\author[2]{Julian-Steffen M\"uller}
\author[3]{Henning Schnoor}
\author[2]{Heribert Vollmer}

\affil[1]{University of Helsinki, Department of Mathematics and Statistics, P.O. Box 68, 00014 Helsinki, Finland}
\affil[2]{Leibniz Universit\"at Hannover, Institut f\"ur Theoretische Informatik, Appelstr.~4, 30167 Hannover, Germany}
\affil[3]{Institut f\"ur Informatik, Christian-Albrechts-Universit\"{a}t zu Kiel, 24098 Kiel, Germany}

\authorrunning{J.\,Kontinen and J.\,S.\,M\"uller and H.\,Schnoor and H.\,Vollmer}

\Copyright{Juha Kontinen and Julian-Steffen M\"uller and Henning Schnoor and Heribert Vollmer}

\newcommand{\neighborhood}[3]{\ensuremath{\mathrm{N}^{#2}_{#3}(#1)}}

\begin{document}

 \maketitle

\begin{abstract}
The famous van Benthem theorem states that modal logic corresponds  exactly to the fragment of first-order logic that is invariant under bisimulation.  In this article we prove an exact analogue of this theorem in the framework of modal dependence logic \MDL and team  semantics. We show that modal team logic \MTL, extending  \MDL by classical negation, captures  exactly the $\FO$-definable  bisimulation invariant properties of Kripke structures and teams. We also compare the expressive power of \MTL to most of the variants and extensions of \MDL recently studied in the area.

\end{abstract}

\section{Introduction}

The concepts of dependence and independence are ubiquitous in many scientific disciplines such as experimental physics, social choice theory, computer science, and cryptography.  Dependence logic \D  \cite{vaananen07} and its so-called team semantics have given rise to a new logical framework in which various notions of dependence and independence can be formalized and studied. Dependence logic extends first-order logic by  dependence atoms
\begin{equation}\label{depatom}
\dep{x_{1},\dots,x_{n-1},x_{n}},
\end{equation}
expressing that the value of the variable $x_{n}$  is functionally dependent on the values of $x_1, \ldots, x_{n-1}$. The formulas of  dependence logic  are evaluated over  \emph{teams}, i.\,e., sets of assignments, and not over single assignments as in first-order logic.

In  \cite{va09}  a modal variant of dependence logic \MDL was introduced. In the modal framework teams are sets of worlds, and a dependence atom 
\begin{equation}\label{modaldepatom}
\dep{p_{1},\dots,p_{n-1},p_{n}}
\end{equation} holds in a team $T$ if there is a Boolean function that determines the value of the propositional variable $p_{n}$ from those of $p_{1},\dots,p_{n-1}$ in all worlds in $T$.   One of the fundamental properties of \MDL (and of dependence logic) is that  its formulas satisfy the so-called downwards closure property:  if $M,T\models\varphi$, and $T'\subseteq T$, then  $M,T'\models\varphi$. Still, the modal framework is very different from the first-order one, e.g.,   dependence atoms between propositional variables can be eliminated with the help of the  classical disjunction $\ovee$  \cite{va09}.  On the other hand, it was recently shown that eliminating  dependence atoms using disjunction causes an exponential blow-up in the formula size, that is,  any formula of $\ML(\ovee)$ logically equivalent to the atom in  \eqref{modaldepatom} is bound to have length exponential in $n$  \cite{DBLP:conf/aiml/HellaLSV14}. The central complexity theoretic questions  regarding \MDL have been solved in \cite{se09,DBLP:journals/sLogica/LohmannV13,eblo12,mu12}.

Extended modal dependence logic, \EMDL, was  introduced in \cite{DBLP:conf/wollic/EbbingHMMVV13}. This extension is defined simply by allowing \ML formulas to appear inside dependence atoms, instead of only propositions. \EMDL can be seen as the first  step towards combining dependencies with temporal reasoning. 
  \EMDL is   strictly more expressive than \MDL but its formulas still have the downwards closure property. In fact,  \EMDL has recently been shown to be  equivalent to the logic $\ML(\ovee)$ \cite{DBLP:conf/wollic/EbbingHMMVV13,DBLP:conf/aiml/HellaLSV14}.
 
In the first-order case, several interesting variants of the dependence atoms have been introduced and studied. The focus has been on independence atoms 
$$(x_1, \dots, x_\ell) \bot_{(y_1, \dots ,y_m)} (z_1, \dots, z_n),$$
and inclusion atoms  
$$(x_1, \dots, x_\ell) \subseteq {(y_1, \dots ,y_\ell)},  $$
which were  introduced in \cite{DBLP:journals/sLogica/GradelV13} and \cite{DBLP:journals/apal/Galliani12}, respectively. 
The intuitive meaning of the independence atom  is that the variables $x_1, \dots, x_\ell$  and $z_1, \dots, z_n$ are independent of each other for any fixed value of  $y_1, \dots ,y_m$, whereas the inclusion atom declares that all values of the tuple $(x_1, \dots, x_\ell)$ appear also as values of $(y_1, \dots ,y_\ell)$.
In \cite{DBLP:conf/aiml/KontinenMSV14} a modal variant, \MIL, of independence logic was introduced. The logic \MIL contains \MDL as a proper sublogic, in particular, its formulas do not in general have the downwards closure property. In  \cite{DBLP:conf/aiml/KontinenMSV14} it was also noted that all \MIL formulas are invariant under  bisimulation when this notion is lifted from single worlds to a relation between sets of words in a natural way.  At the same time (independently) in \cite{DBLP:conf/aiml/HellaLSV14} it was shown that \EMDL and $\ML(\ovee)$ can express  exactly those  properties of Kripke structures and teams that are downwards closed and invariant under $k$-bisimulation for some $k\in \mathbb{N}$.

A famous theorem by Johan van Benthem \cite{vanbenthPHD,BenthemBook} states that modal logic is exactly the fragment of first-order logic that is invariant under  (full) bisimulation. In this paper we study the analogues of this theorem  in the context of team semantics. Our main result shows that an analogue of the van Benthem theorem for team semantics can be  obtained by replacing \ML by \emph{Modal Team Logic} (\MTL). \MTL was introduced in~\cite{JulianThesis} and extends \ML (and \MDL) by classical negation $\classNegation$. More precisely, we show that for any  team property $P$  the following are equivalent:
 \begin{enumerate}[(i)]
  \item\label{enum:main theorem:MTL expressible} There is an \MTL-formula which expresses $P$,
  \item\label{enum:main theorem:FO expressible and bisimulation invariant} there is a first-order formula which expresses $P$ and $P$ is bisimulation-invariant,
  \item\label{enum:main theorem:bounded bisimulation invariant} $P$ is invariant under $k$-bisimulation for some $k$,  
  \item\label{enum:main theorem:bisimulation invariant and local} $P$ is bisimulation-invariant and local.
 \end{enumerate}

 We also study whether all bisimulation invariant properties can be captured by natural variants of \EMDL. We consider extended modal independence and extended modal inclusion logic (\EMIL and \EMINCL, respectively), which are obtained from \EMDL by replacing the dependence atom with the independence (resp.~inclusion) atom. We show that both of these logics fail to capture all bisimulation invariant properties, and therefore in particular are strictly weaker than \MTL. On the other hand, we show that $\EMINCL(\ovee)$ (\EMINCL extended with classical disjunction) is in fact as expressive as $\MTL$, but  the analogously defined $\EMILC$ is strictly weaker. Finally, we show that the extension $\MLFO$ of \ML by all first-order definable generalized dependence atoms (see \cite{DBLP:conf/aiml/KontinenMSV14}) gives rise to a logic that is as well equivalent to \MTL.

\section{Preliminaries}

A \emph{Kripke model} is a tuple $M=(W,R,\pi)$ where $W$ is a nonempty set of worlds, $R\subseteq W\times W$, and $\pi\colon P\rightarrow 2^W$, where $P$ is the set of propositional variables. A \emph{team} of a model $M$ as above is simply a set $T\subseteq W$. The central basic concept underlying V\"a\"an\"anen's modal dependence logic and all its variants is that modal formulas are evaluated not in a world but in a team. This is made precise in the following definitions.
We first recall the usual syntax of modal logic \ML:
$$
\varphi \ddfn p\mid \neg p
\mid (\varphi \wedge \varphi) \mid (\varphi \vee \varphi) 
\mid \Diamond \varphi \mid \Box \varphi,
$$
where $p$ is a  propositional variable. Note that we consider only formulas in negation normal form, i.e., negation appears only in front of atoms. As will become clear from the definition of team semantics of \ML, that we present next, $p$ and $\neg p$ are not dual formulas, consequently \emph{tertium non datur} does not hold in the sense that it is possible that $M,T\nmodels p$ and $M,T\nmodels\neg p$ (however, we still have that $M,T\models p\vee\neg p$ for all models $M$ and teams $T$).
It is worth noting that in  \cite{va09}, the connective $\neg$ is allowed to appear freely in \MDL formulas (with semantics generalizing the atomic negation case of Definition~\ref{def:semanticsML} below, note that classical negation as allowed in \MTL is not allowed in \MDL). The well-known dualities from classical modal logic are also true for  \MDL formulas hence any \ML-formula (even \MDL) can be rewritten in  such a way that $\neg$ only appears in front of propositional variables.
 
\begin{definition}\label{def:semanticsML}
Let $M=(W,R,\pi)$ be a Kripke model, let $T\subseteq W$ be a team, and let $\varphi$ be an \ML-formula. We define when $M,T\models\varphi$ holds inductively:
  \begin{itemize}
  \item If $\varphi=p$, then $M,T\models\varphi$ iff $T\subseteq\pi(p)$,
  \item If $\varphi=\neg p$, then $M,T\models\varphi$ iff $T\cap\pi(p)=\emptyset$,
  \item If $\varphi=\psi\vee\chi$ for some formulas $\psi$ and $\chi$, then $M,T\models\varphi$ iff $T=T_1\cup T_2$ with $M,T_1\models\psi$ and $M,T_2\models\chi$,
  \item If $\varphi=\psi\wedge\chi$ for some formulas $\psi$ and $\chi$, then $M,T\models\varphi$ iff $M,T\models\psi$ and $M,T\models\chi$,
  \item If $\varphi=\Diamond\psi$ for some formula $\psi$, then $M,T\models\varphi$ iff there is some team $T'$ of $M$ such that $M,T'\models\psi$ and 
   \begin{inparaenum}
    \item for each $w\in T$, there is some $w'\in T'$ with $(w,w')\in R$, and 
    \item for each $w'\in T'$, there is some $w\in T$ with $(w,w')\in R$.
   \end{inparaenum}
  \item If $\varphi=\Box\psi$ for some formula $\psi$, then $M,T\models\varphi$ iff $M,T'\models\psi$, where $T'$ is the set $\set{w'\in M\ \vert\ (w,w')\in R\mathtext{ for some }w\in T}$.
 \end{itemize}
 \end{definition}

Analogously to the first-order setting,  $\ML$-formulas satisfy  the following \emph{flatness} property~\cite{vaananen07}. Here, the notation $M,w\models\varphi$ in item~\ref{enum:flatness individual worlds} refers to the standard semantics of  modal logic (without teams).
\begin{proposition}\label{proposition:flatness}
 Let $M$ be a Kripke model and $T$ a team of $M$. Let $\varphi$ be an $\ML$-formula. Then the following are equivalent:
\begin{enumerate}
\item $M,T\models\varphi$,
\item $M,\set w\models\varphi$ for each $w\in T$,
\item\label{enum:flatness individual worlds} $M,w\models\varphi$ for each $w\in T$.
\end{enumerate}
\end{proposition}

 \emph{Modal team logic} extends \ML by a second type of negation, denoted by $\classNegation$, and interpreted just as classical negation. The syntax is formally given as follows:
$$
\varphi \ddfn p\mid \neg p \mid \classNegation\varphi
\mid (\varphi \wedge \varphi) \mid (\varphi \vee \varphi) 
\mid \Diamond \varphi \mid \Box \varphi,
$$
where $p$ is a  propositional variable. 
The semantics of \MTL is defined by extending Def.~\ref{def:semanticsML} by the following clause:
 \begin{itemize}
  \item If $\varphi=\classNegation\psi$ for some formula $\psi$, then $M,T\models\varphi$ iff  $M,T\nmodels\psi$.
 \end{itemize}
We note that usually (see \cite{JulianThesis}), \MTL also contains dependence atoms; however since these atoms can be expressed in \MTL  we omit them  in the syntax (see Proposition \ref{defatoms} below).  The classical disjunction $\ovee$ (in some other context also referred to as \emph{intuitionistic disjunction}) is also readily expressed in \MTL: 
 $\varphi\ovee\psi$ is logically equivalent to $\classNegation(\classNegation\varphi\wedge\classNegation\psi)$.

For an \ML formula $\varphi$, we let $\varphi^{dual}$ denote the formula that is obtained by transforming $\neg\varphi $ to negation normal form. Now by Proposition \ref{proposition:flatness}  it follows that 
\[ M,T\models\varphi^{dual} \textrm{ iff } M,w\not \models \varphi \textrm{ for all }w\in T, \]
hence 
 $M,T\models\classNegation\psi^{dual}$ if and only if there is some $w\in T$ with $M, w\models\psi$. We therefore often write $\existsWorld\psi$ instead of $\classNegation\psi^{dual}$.
 Note that $\existsWorld$ is not a global operator stating existence of a world anywhere in the model, but $\existsWorld$ is evaluated in the current team. It is easy to see (and follows from Proposition~\ref{proposition:team bisimulation implies equivalence}) that a global ``exists'' operator cannot be expressed in \MTL.

The next proposition shows that dependence atoms can be easily expressed in \MTL.

\begin{proposition}\label{defatoms}
The dependence atom  \eqref{modaldepatom} can be expressed in \MTL by a formula that has length polynomial in $n$. 
\end{proposition}

\begin{proof} Note first that, analogously to the first-order case \cite{DBLP:journals/synthese/AbramskyV09}, \eqref{modaldepatom} is logically equivalent with 
\[ (\bigwedge_{1\le i\le n-1} \dep{p_i} )  \rightarrow \dep{p_n},      \]
where  $\rightarrow $ is the so-called intuitionistic implication with the following semantics:
\[M,T\models \varphi \rightarrow \psi \textrm{ iff for all $T'\subseteq T$: if $M,T'\models \varphi$ then $M,T'\models \psi$}. \]
The connective $\rightarrow$ has a short logically equivalent definition in \MTL (see \cite{JulianThesis}): $\varphi\rightarrow\psi$ is equivalent to $(\sim\varphi\ovee\psi)\otimes\bot$, where $\otimes$ is the dual of  $\vee$, i.e., $\varphi\otimes\psi:= \classNegation(\classNegation\varphi\vee\classNegation\psi)$, and $\bot$ is a shorthand for the formula $p_0 \wedge \neg p_0$. Finally, $\dep{p_i}$ can be written as $p_i\ovee \neg p_i$, hence the claim follows.
\end{proof}

The intuitionistic implication used in the proof above has been studied in the modal team semantics context in \cite{Yangthesis}.

We now introduce the central concept of bisimulation~\cite{PARK-CONCURRENCY-AUTOMATA-INFINITE-SEQUENCES-TCS-1981,BenthemBook}. Intuitively, two \emph{pointed models} (i.e., pairs of models and worlds from the model) $(M_1,w_1)$ and $(M_2,w_2)$ are bisimilar, if they are indistinguishable from the point of view of modal logic. The notion of $k$-bisimilarity introduced below corresponds to indistinguishability by formulas with modal depth up to $k$: For a formula $\varphi$ in any of the logics considered in this paper, the \emph{modal depth} of $\varphi$, denoted with $\md\varphi$, is the maximal nesting degree of modal operators (i.e., $\Box$ and $\Diamond$) in $\varphi$.

\begin{definition}
 Let $M_1=(W_1,R_1,\pi_1)$ and $M_2=(W_2,R_2,\pi_2)$ be Kripke models. We define inductively what it means for states $w_1\in W_1$ and $w_2\in W_2$ to be $k$-bisimilar, for some $k\in\mathbb N$, written as $(M_1,w_1)\bisim k (M_2,w_2)$.
 
 \begin{itemize}
  \item $(M_1,w_1)\bisim 0(M_2,w_2)$ holds if for each propositional variable $p$, we have that $M_1,w_1\models p$ if and only if $M_2,w_2\models p$.
  \item $(M_1,w_1)\bisim{k+1}(M_2,w_2)$ holds if the following three conditions are satisfied:
  \begin{enumerate}
   \item $(M_1,w_1)\bisim{0}(M_2,w_2)$,
   \item for each successor $w_1'$ of $w_1$ in $M_1$, there is a successor $w_2'$ of $w_2$ in $M_2$ such that $(M_1,w_1')\bisim k(M_2,w_2')$ (\emph{forward} condition),
   \item for each successor $w_2'$ of $w_2$ in $M_2$, there is a successor $w_1'$ of $w_1$ in $M_1$ such that $(M_1,w_1')\bisim k(M_2,w_2')$ (\emph{backward} condition).
  \end{enumerate}
 \end{itemize}
\end{definition}

Full bisimilarity is defined similarly: Pointed models $(M_1,w_1)$ and $(M_2,w_2)$ are \emph{bisimilar} if there is a relation $Z\subseteq W_1\times W_2$ such that $(w_1,w_2)\in Z$, and for all $(w_1,w_2)\in Z$, we have that $w_1$ and $w_2$ satisfy the same propositional variables, and for each successor $w_1'$ of $w_1$ in $M_1$, there is a successor $w_2'$ of $w_2$ in $M_2$ with $(w_1',w_2')\in Z$ (forward condition), and analogously for each successor $w_2'$ of $w_2$ in $M_2$, there is a successor $w_1'$ of $w_1$ in $M_1$ with $(w_1',w_2')\in Z$ (back condition). In this case we simply say that $M_1,w_1$ and $M_2,w_2$ are \emph{bisimilar}. It is easy to see that bisimilarity implies $k$-bisimilarity for each $k$. 

\begin{definition}
 Let $M_1=(W_1,R_1,\pi_1)$ and $M_2=(W_2,R_2,\pi_2)$ be Kripke models, and let $w_1\in W_1$, $w_2\in W_2$. Then $(M_1,w_1)$ and $(M_2,w_2)$ are $k$-equivalent for some $k\in\mathbb N$, written $(M_1,w_1)\equiv_k (M_2,w_2)$ if for each modal formula $\varphi$ with $\md\varphi\leq k$, we have that $M_1,w_1\models\varphi$ if and only if $M_2,w_2\models\varphi$.
\end{definition}

Again, we simply write $w_1\equiv_{k} w_2$ if the models $M_1$ and $M_2$ are clear from the context. As mentioned above, $k$-bisimilarity and $k$-equivalence coincide. The following result is standard (see, e.g., \cite{BlackburnDeRijkeVenama-MODAL-LOGIC-BOOK-2001}):

\begin{proposition}\label{prop:bisimulation and equivalence}
 Let $M_1=(W_1,R_1,\pi_1)$ and $M_2=(W_2,R_2,\pi_2)$ be Kripke models, and let $w_1\in W_1$, $w_2\in W_2$. Then $(M_1,w_1)\bisim k (M_2,w_2)$ if and only if $(M_1,w_1)\equiv_k (M_2,w_2)$. 
\end{proposition}

For \MTL and more generally logics with team semantics, the above notion of bisimulation can be lifted to teams. The following definition is a natural adaptation of $k$-bisimilarity to the team setting:

\begin{definition}
 Let $M_1=(W_1,R_1,\pi_1)$ and $M_2=(W_2,R_2,\pi_2)$ be Kripke models, let $T_1$ and $T_2$ be teams of $M_1$ and $M_2$. Then $(M_1,T_1)$ and $(M_2,T_2)$ are $k$-bisimilar, written as $M_1,T_1\bisim{k}M_2,T_2$ if the following holds:
 \begin{itemize}
  \item for each $w_1\in T_1$, there is some $w_2\in T_2$ such that $(M_1,w_1)\bisim{k} (M_2,w_2)$,
  \item for each $w_2\in T_2$, there is some $w_1\in T_1$ such that $(M_1,w_1)\bisim{k} (M_2,w_2)$.
 \end{itemize}
\end{definition}

Full bisimilarity on the team level is defined analogously. In this case we again say that $(M_1,T_1)$ and $(M_2,T_2)$ are bisimilar, and write $M_1,T_1\bisim{}M_2,T_2$, if there is a relation $Z\subseteq W_1\times W_2$ satisfying the forward and backward conditions as above, and which additionally satisfies that for each $w_1\in T_1$, there is some $w_2\in T_2$ with $(w_1,w_2)\in Z$, and for each $w_2\in T_2$, there is some $w_1\in T_1$ with $(w_1,w_2)\in Z$. This notion of team-bisimilarity was first introduced in~\cite{DBLP:conf/aiml/KontinenMSV14} and \cite{DBLP:conf/aiml/HellaLSV14}. The following result is easily proved by induction on the formula length:

\begin{restatable}{proposition}{propositionteambisimulationimpliesequivalence}\label{proposition:team bisimulation implies equivalence}
 Let $M_1$ and $M_2$ be Kripke models, let $T_1$ and $T_2$ be teams of $M_1$ and of $M_2$. Then 
 \begin{enumerate}
  \item If $(M_1,T_1)\bisim k(M_2,T_2)$, then for each formula $\varphi\in\MTL$ with $\md\varphi\leq k$, we have that $M_1,T_1\models\varphi$ if and only if $M_2,T_2\models\varphi$.
  \item If $(M_1,T_1)\bisim{}(M_2,T_2)$, then for each formula $\varphi\in\MTL$, we have that $M_1,T_1\models\varphi$ if and only if $M_2,T_2\models\varphi$.
 \end{enumerate}
\end{restatable}

\begin{proof}
 It suffices to show the first claim, since $(M_1,T_1)\bisim{}(M_2,T_2)$ implies $(M_1,T_1)\bisim k(M_2,T_2)$ for every $k\in\mathbb N$. We show the first claim by induction over the construction of $\varphi$. In this proof, we omit the models in the notation $\bisim{k}$ (for the entire proof, the model on the left-hand side of $\bisim k$ is always $M_1$, and the model on the right-hand side is always $M_2$). Clearly, it suffices to show that if $M_1,T_1\models\varphi$, then $M_2,T_2\models\varphi$. Hence assume that $M_1,T_1\models\varphi$
 
 \begin{itemize}
  \item Assume that $\varphi=p$ for a variable $p$ (the case $\varphi=\neg p$ is analogous).  Since $M_1,T_1\models\varphi$, we know that $M_1,w_1\models p$ for each $w_1\in T_1$. Let $w_2\in T_2$. Since $T_1\bisim kT_2$, there is some $w_1\in T_1$ with $w_1\bisim kw_2$. By the above, we know that $M_1,w_1\models p$, and by definition of $\bisim k$, it follows that $M_2,w_2\models p$. Hence $M_2,T_2\models\varphi$ as required.
  \item Assume that $\varphi=\psi\wedge\chi$. Since $M_1\models\varphi$, we know that $M_1,T_1\models\psi$ and $M_1,T_1\models\chi$. By induction, it follows that $M_2,T_2\models\psi$ and $M_2,T_2\models\chi$. Therefore, $M_2,T_2\models\varphi$ as required.
  \item Assume that $\varphi=\psi\vee\chi$. Then $T_1=T_1'\cup T_1''$ with $M_1,T_1'\models\psi$ and $M_1,T_1''\models\chi$.
  
  Let $T_2'=\set{w_2\in T_2\ \vert\ \mathtext{there is some }w_1\in T_1'\mathtext{ with }w_1\bisim kw_2}$ and \linebreak $T_2''=\set{w_2\in T_2\ \vert\ \mathtext{there is some }w_1\in T_1''\mathtext{ with }w_1\bisim kw_2}$. It can easily be verified that $T_2'\cup T_2''=T_2$ and $T_1'\bisim kT_2'$ and $T_1''\bisim kT_2''$. By induction, it follows that $M_2,T_2'\models\psi$ and $M_2,T_2''\models\chi$, and hence $M_2,T_2\models\varphi$ as required.
    
  \item Assume that $\varphi=\classNegation\psi$. Then $M_1,T_1\nmodels\psi$. Assume indirectly that $M_2,T_2\models\psi$. Due to induction, it then follows that $M_1,T_1\models\psi$, a contradiction.
  
  \item Assume that $\varphi=\Box\psi$, let $T_1'=\set{w_1'\ \vert\ (w_1,w_1')\in R_1\mathtext{ for some }w_1\in T_1}$, and $T_2'=\set{w_2'\ \vert\ (w_2,w_2')\in R_1\mathtext{ for some }w_2\in T_2}$. From $M_1,T_1\models\varphi$, we know that $M_1,T_1'\models\psi$, and we need to show that $M_2,T_2'\models\psi$. Due to induction, and since $\md\psi=\md\varphi-1\leq k-1$, it suffices to show that $T_1'\bisim{k-1}T_2'$.
  
  Hence let $w_1'\in T_1'$. Then $w_1'$ is the $R_1$-successor of some $w_1\in T_1$. Since $T_1\bisim k T_2$, there is some $w_2\in T_2$ with $w_1\bisim kw_2$. In particular, there is some $w_2'$ such that $w_2'$ is an $R_2$-successor of $w_2$, and $w_1'\bisim kw_2'$. The other direction is symmetric.

  \item Assume that $\varphi=\Diamond\psi$. Since $M_1,T_1\models\varphi$, there is a team $T_1'$ of $M_1$ such that for each $w_1\in T_1$, there is some $w_1'\in T_1'$ and $(w_1,w_1')\in R_1$. We define the team $T_2'$ as follows:
  
  $T_2'=\{w_2'\in M_2\ \vert\ w_2'\mathtext{ has an }R_2\mathtext{-predecessor in }T_2$\\ \phantom{a} \hfill$\mathtext{ and there is }w_1\in T_1,w_1'\in T_1'\mathtext{ with }(w_1,w_1')\in R_1\mathtext{ and }w_1'\bisim{k-1}w_w'.\}$
  
  By definition, $T_2'$ only contains $R_2$-successors of worlds in $T_2$, and for each world $w_2'$ in $T_2'$, there is some $w_1'\in T_1'$ with $w_1'\bisim kw_2'$. To complete the proof, it therefore remains to show that:
  \begin{enumerate}
   \item For each $w_2\in T_2$, there is some $R_2$-successor $w_2'$ of $w_2$ with $w_2'\in T_2$,
   \item for each $w_1'\in T_1'$, there is some $w_2'\in T_2'$ with $w_1'\bisim kw_2'$.
  \end{enumerate}

  We now prove these claims:
  
  \begin{enumerate}
   \item Let $w_2\in T_2$. Then there is some $w_1\in T_1$ with $w_1\bisim kw_2$. By choice of $T_1'$,there is some $w_1'\in T_1'$ with $(w_1,w_1')\in R_1$. Since $w_1\bisim kw_2$, there is an $R_2$-successor $w_2'$ of $T_2'$ with $w_1'\bisim{k-1}w_2'$. Hence by choice of $T_2'$, we know that $w_2'\in T_2'$, and $w_2'$ is an $R_2$-successor of $w_2$.
   \item Let $w_1'\in T_1'$. Then $w_1'$ is the $R_1$-successor of some $w_1\in T_1$. Since $T_1\bisim kT_2$, there is some $w_2\in T_2$ with $w_1\bisim kw_2$. By choice of $T_1'$, there is an $R_1$-successor $w_1'$ of $w_1$ with $w_1'\in T_1'$. Since $w_1\bisim kw_2$, there is some $R_2$-successor $w_2'$ of $w_2$ with $w_1'\bisim{k-1}w_2'$. Hence, by choice of $T_2'$, we have $w_2'\in T_2'$ as required.
  \end{enumerate}
 \end{itemize}
\end{proof}

The expressive power of classical modal logic (i.e., without team semantics) can be characterized by bisimulations. In particular, for every pointed model $(M,w)$, there is a modal formula of modal depth $k$ that exactly characterizes $(M,w)$ up to $k$-bisimulation.

In the following, we restrict ourselves to a finite set of propositional variables.

\section{Main Result: Expressiveness of \MTL}\label{sect:MTL expressiveness}

In this section, we study the expressive power of \MTL. As usual, we measure the expressive power of  a logic by the set of properties expressible in it.

\begin{definition}
A \emph{team property} is a class of pairs $(M,T)$ where $M$ is a Kripke model and $T\neq\emptyset$ a team of $M$. For an \MTL-formula $\varphi$, we say that $\varphi$ \emph{expresses} the property $\set{(M,T)\ \vert\ M,T\models\varphi}$.
\end{definition}

Note that most variants of modal dependence logic have the \emph{empty team property}, i.e., for all $\varphi\in\EMINCL$ and all Kripke structures $M$, we have $M,\emptyset \models \varphi$, which obviously does not hold for $\MTL$. 
However, it immediately follows from the bisimulation invariance of \MTL that for every \MTL formula $\varphi$ one of the two possibilities hold:
\begin{itemize}
	\item For all Kripke structures $M$, $M,\emptyset \models \varphi$.
	\item For all Kripke structures $M$, $M,\emptyset \not\models \varphi$.
\end{itemize}
For this reason we exclude the empty team in the statement of our results below, but we note that by the remarks above all results cover also the empty team.

\begin{definition}
 Let $P$ be a team property. Then $P$ is \emph{invariant under $k$-bi\-si\-mu\-la\-ti\-on} if for each pair of Kripke models $M_1$ and $M_2$ and teams $T_1$ and $T_2$ with $(M_1,T_1)\bisim k(M_2,T_2)$ and $(M_1,T_1)\in P$, it follows that $(M_2,T_2)\in P$.
\end{definition}

We introduce some (standard) notation. In a model $M$, the \emph{distance} between two worlds $w_1$ and $w_2$ of $M$ is the length of a shortest path from $w_1$ to $w_2$ (the distance is infinite if there is no such path). For a world $w$ of a model $M$ and a natural number $d$, the \emph{$d$-neighborhood of $w$ in $M$}, denoted $\neighborhood wdM$, is the set of all worlds $w'$ of $M$ such that the distance from $w$ to $w'$ is at most $d$. For a team $T$, with $\neighborhood TdM$ we denote the set $\cup_{w\in T}\neighborhood wdM$. We often identify $\neighborhood TdM$ and the model obtained from $M$ by restriction to the worlds in $\neighborhood TdM$.

\begin{definition}
 A team property $P$ is \emph{$d$-local} for some $d\in\mathbb N$ if for all models $M$ and teams $T$, we have $$(M,T)\in P\mathtext{ if and only if }(\neighborhood TdM,T)\in P.$$ We say that $P$ is \emph{local}, if $P$ is $d$-local for some $d\in\mathbb N$.
\end{definition}

Since our main result establishes a connection between team properties definable in MTL and team properties definable in first-order logic, we also define what it means for a team property to be expressed by a first-order formula. For a finite set of propositional variables $\propVar$, we define $\sigma_{\propVar}$ as the first-order signature containing a binary relational symbol $E$ (for the edges in our model), a unary relational symbol $T$ (for representing a team), and, for each variable $x\in\propVar$, a unary relational symbol $W_x$ (representing the worlds in which $x$ is true). Kripke models $M$ with teams $T$ (where we only consider variables in $\propVar$) directly correspond to $\sigma_{\propVar}$ structures: A model $M=(W,R,\pi)$ and a team $T$ uniquely determines the $\sigma_{\propVar}$-structure $\foStructure MT$ with universe $W$ and the obvious interpretations of the symbols in $\sigma_{\propVar}$. 

We therefore say that a first-order formula $\varphi$ over the signature $\sigma_{\propVar}$ expresses a team property $P$, if for all models $M$ with a team $T$, we have that $(M,T)\in P$ if and only if $\foStructure MT\models\varphi$. We can now state the main result of this paper:

\begin{theorem}\label{theorem:MTL expressiveness characterization}
 Let $P$ be a team property. Then the following are equivalent:
 \begin{enumerate}[(i)]
  \item\label{enum:main theorem:MTL expressible} There is an \MTL-formula which expresses $P$,
  \item\label{enum:main theorem:FO expressible and bisimulation invariant} there is a first-order formula which expresses $P$ and $P$ is bisimulation-invariant,
  \item\label{enum:main theorem:bounded bisimulation invariant} $P$ is invariant under $k$-bisimulation for some $k$,  
  \item\label{enum:main theorem:bisimulation invariant and local} $P$ is bisimulation-invariant and local.
 \end{enumerate}
\end{theorem}

This result characterizes the expressive power of \MTL\ in several ways. The equivalence of points~\ref{enum:main theorem:MTL expressible} and~\ref{enum:main theorem:FO expressible and bisimulation invariant} is a natural analog to the classic van Benthem theorem which states that standard modal logic directly corresponds to the bisimulation-invariant fragment of first-order logic. It is easy to see that characterizations corresponding to items~\ref{enum:main theorem:bounded bisimulation invariant} and~\ref{enum:main theorem:bisimulation invariant and local} also hold in the classical setting. Our result therefore shows that \MTL\ plays the same role for team-based modal logics as \ML\ does for standard modal logic.

The connection between our result and van Benthem's Theorem~\cite{vanbenthPHD,BenthemBook} is also worth discussing. Essentially, van Benthem's Theorem is the same result as ours, where ``\MTL'' is replaced by ``\ML'' and properties of pointed models (i.e., singleton teams) are considered. In \ML, classical negation is of course freely available; however the property of a team being a singleton is clearly not invariant under bisimulation---but the property of a team having only one element \emph{up to bisimulation} is. It therefore follows that each property of singleton teams that is invariant under bisimulation and that can be expressed in \MTL\ can already be expressed in \ML.

The remainder of Section~\ref{sect:MTL expressiveness} is devoted to the proof of Theorem~\ref{theorem:MTL expressiveness characterization}. The proof relies on various formulas that characterize pointed models, teams of pointed models, or team properties up to $k$-bisimulation, for some $k\in\mathbb N$. In Table~\ref{fig:formulas overview}, we summarize the notation used in the following and explain the intuitive meaning of these formulas.

\subsection{Expressing Properties in MTL and Hintikka Formulas}

We start with a natural characterization of the semantics of splitjunction $\vee$ for \ML-formulas.

\begin{restatable}{proposition}{lemmabigsplitjunctioncharacterization}\label{lemma:big splitjunction characterization}
 Let $S$ be a non-empty finite set of \ML-formulas, let $M$ be a model and $T$ a team. Then
 $M,T\models\bigvee_{\varphi\in S}\varphi$ if and only if for each world $w\in T$, there is a formula $\varphi\in S$ with $M,\set w\models\varphi$.
\end{restatable}

\begin{proof}
Let $M$ be a model and $T$ a team. For $\card S=1$, the claim follows from Proposition~\ref{proposition:flatness}. Therefore assume that the lemma holds for $S'=\set{\varphi_1,\dots,\varphi_{n-1}}$, and consider a set $S=S'\cup\set{\varphi_n}$. Then $\bigvee_{\varphi\in S}\varphi=\bigvee_{\varphi\in S'}\vee\varphi_n$.
 
First assume that $(M,T)\models\bigvee_{\varphi\in S}$. Then $(M,T)\models\bigvee_{\varphi\in S'}\varphi\vee\varphi_n$. Therefore, $T=T_1\cup T_2$ with
 
 \begin{itemize}
  \item $M,T_1\models\bigvee_{\varphi\in S'}\varphi$,
  \item $M,T_2\models\varphi_n$.
 \end{itemize}

 Due to the induction assumption, we know that for every world $w\in T_1$, there is a formula $\varphi\in S'\subseteq S$ with $M,w\models\varphi$. Due to Proposition~\ref{proposition:flatness},  we know that each world $w\in T_2$ satisfies the formula $\varphi_n\in S$, therefore each world $w\in T$ satisfies some formula $\varphi\in S$ as claimed.
 
 For the converse, assume that for each $w\in T$, there is some formula $\varphi\in S$ with $M,w\models\varphi$. Let $T_1=\set{w\in T\ \vert\ M,w\models\varphi_i\mathtext{ for some }i\leq n-1}$, and let $T_2=\set{w\in T\ \vert\ M,w\models\varphi_n}$. By the prerequisites we know that $T=T_1\cup T_2$. Due to induction, we have that $M,T_1\models\bigvee_{\varphi\in S'}\varphi$, and, due to Proposition~\ref{proposition:flatness}, we know that $M,T_2\models\varphi_n$. It therefore follows that $M,T\models\bigvee_{\varphi\in S'}\varphi\vee\varphi_n=\bigvee_{\varphi\in S}\varphi$ as claimed.
\end{proof}

\begin{table}
 \begin{tabular}{lp{9.5cm}l}
  \textbf{Formula} & \textbf{Intuition} & Defined in \\ \hline 
  $\hintikkaPointedModel Mwk$ & Characterizes the pointed model $(M,w)$ up to $k$-bisimilarity  & Theorem~\ref{theorem:hintikka formulas} \\
  $\hintikkaFormulasLevel k$ & All formulas of the form $\hintikkaPointedModel Mwk$ (this is a finite set) & Definition~\ref{definition:hintikkaFormulasLevel} \\
  $\hintikkaSetOfPointedModelsInTeam MTk$ & Formulas 
  characterizing pointed models $(M,w)$, where $w\in T$, up to $k$-bisimilarity (this is a finite set) & Definition~\ref{definition:hintikkaSetOfPointedModelsInTeam}\\
  $\hintikkaModelWithTeam MTk$ & Formula characterizing model $M$ with $T$ up to $k$-bisimilarity & Definition~\ref{definition:hintikkaModelWithTeam}
   \end{tabular}\\[.5ex]
\caption{Formulas and sets of formulas used in the proof of Theorem~\ref{theorem:MTL expressiveness characterization}}\label{fig:formulas overview}
\end{table}

The following result is standard:

\begin{theorem}\label{theorem:hintikka formulas}{\rm\bfseries\cite[Theorem 32]{GorankoOtto06}}-
 For each pointed Kripke model $(M,w)$ and each natural number $k$, there is a \emph{Hintikka formula} (or characteristic formula) $\hintikkaPointedModel Mwk\in\ML$ with $\md{\hintikkaPointedModel Mwk}=k$ such that for each pointed model $(M',w')$, the following are equivalent:
 \begin{enumerate}
  \item $M',w'\models\hintikkaPointedModel Mwk$,
  \item $(M,w)\bisim k(M',w')$.
 \end{enumerate}
\end{theorem}

Clearly, we can choose the Hintikka formulas such that $\hintikkaPointedModel Mwk$ is uniquely determined by the bisimilarity type of $(M,w)$. This implies that for $k$-bisimilar pointed models $(M_1,w_1)$ and $(M_2,w_2)$, the formulas $\hintikkaPointedModel{M_1}{w_1}k$ and $\hintikkaPointedModel{M_2}{w_2}k$ are identical. 

It it clear that Theorem~\ref{theorem:hintikka formulas} does not hold for an infinite set of propositional symbols, since a finite formula can only specify the values of finitely many variables.

We now define the set of all Hintikka formulas that will appear in our later constructions. Informally, $\hintikkaFormulasLevel k$ is the set of all Hintikka formulas characterizing models up to $k$-bisimilarity:

\begin{definition}\label{definition:hintikkaFormulasLevel}
For $k\in \mathbb{N}$, the set $\hintikkaFormulasLevel k$ is defined as 
 $$\hintikkaFormulasLevel k=\set{\hintikkaPointedModel Mwk\ \vert\ (M,w)\mathtext{ is a pointed Kripke model}}.$$
\end{definition}

An important observation is that $\hintikkaFormulasLevel k$ is a finite set: This follows since above, we chose the representatives $\hintikkaPointedModel Mwk$ to be identical for $k$-bisimilar pointed models, and since there are only finitely many pointed models up to $k$-bisimulation. Since $\hintikkaFormulasLevel k$ is finite, we can in the following freely use disjunctions over arbitrary subsets of $\hintikkaFormulasLevel k$ and still obtain a finite formula. We will make extensive use of this fact in the remainder of Section~\ref{sect:MTL expressiveness}, often without reference.

Our next definition is used to characterize a team, again up to $k$-bisimulation. Since teams are sets of worlds, we use sets of formulas to characterize teams in the natural way, by choosing, for each world in the team, one formula that characterizes it.

\begin{definition}\label{definition:hintikkaSetOfPointedModelsInTeam}
 For a model $M$ and a team $T$, let
 $$\hintikkaSetOfPointedModelsInTeam MTk=\set{\varphi\in\hintikkaFormulasLevel k\ \vert\ \mathtext{there is some }w\in T\mathtext{ with }M,w\models\varphi}.$$
\end{definition}

Since $\hintikkaSetOfPointedModelsInTeam MTk\subseteq\hintikkaFormulasLevel k$, it follows that $\hintikkaSetOfPointedModelsInTeam MTk$ is finite as well. In fact, it is easy to see that $\card{\hintikkaSetOfPointedModelsInTeam MTk}$ is exactly the number of $k$-bisimilarity types in $T$, i.e., the size of a maximal subset of $T$ containing only worlds such that the resulting pointed models are pairwise non-$k$-bisimilar.

We now combine the formulas from $\hintikkaSetOfPointedModelsInTeam MTk$ to be able to characterize $M$ and $T$ (up to $k$-bisimulation) by a single formula:

\begin{definition}\label{definition:hintikkaModelWithTeam}
 For a model $M$ with a team $T\neq\emptyset$, let $$\hintikkaModelWithTeam MTk=\left(\bigwedge_{\varphi\in\hintikkaSetOfPointedModelsInTeam MTk}\existsWorld\varphi\right)\wedge\left(\bigvee_{\varphi\in\hintikkaSetOfPointedModelsInTeam MTk}\varphi\right).$$
\end{definition}

Intuitively, the formula $\hintikkaModelWithTeam MTk$ expresses that in a model $M'$ and $T'$ with $M',T'\models\hintikkaModelWithTeam MTk$, for each world $w\in T$ there must be some $w'\in T'$ such that $(M,w)\bisim k(M',w')$, and conversely, for each $w'\in T'$, there must  be some $w\in T$ with $(M,w)\bisim k(M',w')$, which then implies that $(M,T)$ and $(M',T')$ are indeed $k$-bisimilar.

From the above, it follows that $\hintikkaModelWithTeam MTk$ is a finite \MTL-formula. Therefore, with the above intuition, it follows that $\hintikkaModelWithTeam MTk$ expresses $k$-bisimilarity with $(M,T)$. 

\begin{restatable}{proposition}{lemmateambisimulationcharacterizationwithhintikkaformulas}
\label{proposition:team bisimulation characterization with hintikka formulas}
 Let $M_1,M_2$ be Kripke models with teams nonempty $T_1,T_2$. Then the following are equivalent:
 \begin{itemize}
  \item $(M_1,T_1)\bisim k (M_2,T_2)$
  \item $M_1,T_1\models\hintikkaModelWithTeam{M_2}{T_2}k$.
 \end{itemize}
\end{restatable}

\begin{proof}
 First assume that $(M_1,T_1)\bisim k (M_2,T_2)$. To see that $M_1,T_1\models\bigwedge_{\varphi\in\hintikkaSetOfPointedModelsInTeam{M_2}{T_2}k}\existsWorld\varphi$, let $\varphi\in\hintikkaSetOfPointedModelsInTeam{M_2}{T_2}k$. By definition, $\varphi$ is an \ML-formula with $\md\varphi\leq k$ and there is some world $w_2\in T_2$ such that $M_2,w_2\models\varphi$. Due to the bisimulation condition, we know that there is a world $w_1\in T_1$ with $(M_1,w_1)\bisim k(M_2,w_2)$. From Proposition~\ref{prop:bisimulation and equivalence}, it follows that $M_1,w_1\models\varphi$. Since $w_1\in T_1$, it follows that $M_1,T_1\models\existsWorld\varphi$. Since this is true for each $\varphi\in\hintikkaSetOfPointedModelsInTeam{M_2}{T_2}k$, it follows that $M_1,T_1\models\bigwedge_{\varphi\in\hintikkaSetOfPointedModelsInTeam{M_2}{T_2}k}\existsWorld\varphi$.
 
 It remains to show that $M_1,T_1\models\bigvee_{\varphi\in\hintikkaSetOfPointedModelsInTeam{M_2}{T_2}k}\varphi$. Due to Lemma~\ref{lemma:big splitjunction characterization}, it suffices to show that for each world $w_1\in T_1$, there is some formula $\varphi\in\hintikkaSetOfPointedModelsInTeam{M_2}{T_2}k$ with $M_1,w_1\models\varphi$. Therefore, let $w_1\in T_1$. Due to the bisimulation condition, we know that there is some $w_2\in T_2$ with $(M_1,w_1)\bisim k (M_2,w_2)$. Due to the definition of $\hintikkaSetOfPointedModelsInTeam{M_2}{T_2}k$, we know that for the formula $\varphi:=\Phi^k_{M_2,w_2}$ we have that $\varphi\in\hintikkaSetOfPointedModelsInTeam{M_2}{T_2}k$ and $M_2,w_2\models\varphi$. Since $(M_1,w_1)\bisim k(M_2,w_2)$, it follows that $M_1,w_1\models\varphi$, and hence for each $w_1\in T_1$ there is a formula $\varphi$ as required.
 
 For the converse, assume that $M_1,T_1\models\hintikkaModelWithTeam{M_2}{T_2}k$.
 In particular, it follows that $M_1,T_1\models \bigvee_{\varphi\in\hintikkaSetOfPointedModelsInTeam{M_2}{T_2}k}\varphi$. Due to Lemma~\ref{lemma:big splitjunction characterization}, it follows that for each $w_1\in T_1$, there is some formula $\varphi\in\hintikkaSetOfPointedModelsInTeam{M_2}{T_2}k$ with $M_1,w_1\models\varphi$. Due to the definition of $\hintikkaSetOfPointedModelsInTeam{M_2}{T_2}k$, we know that $\varphi\in\hintikkaFormulasLevel k$, and that there is some $w_2\in T_2$ with $M_2,w_2\models\varphi$. From Theorem~\ref{theorem:hintikka formulas}, it therefore follows that for each $w_1\in T_1$, there is some $w_2\in T_2$ with $(M_1,w_1)\bisim k(M_2,w_2)$.
 
 Let $w_2$ be a world in $T_2$, and let $\varphi:=\Phi^k_{M_2,w_2}$. It then follows that $\varphi\in\hintikkaSetOfPointedModelsInTeam{M_2}{T_2}k$. Since $M_1,T_1\models\bigwedge_{\varphi\in\hintikkaSetOfPointedModelsInTeam{M_2}{T_2}k}\existsWorld\varphi$, it follows from the definition of the operator $\existsWorld$ that there is some $w_1\in T_1$ with $M_1,w_1\models\varphi$. Due to the choice of $\varphi$, Theorem~\ref{theorem:hintikka formulas} implies that $(M_1,w_1)\bisim k(M_2,w_2)$ as required.
\end{proof}

\subsection{Proof of Theorem~\ref{theorem:MTL expressiveness characterization}}

In this section, we prove our main result, Theorem~\ref{theorem:MTL expressiveness characterization}.

\subsubsection{Proof of equivalence~\ref{theorem:MTL expressiveness characterization}.(\ref{enum:main theorem:MTL expressible}) $\leftrightarrow$ \ref{theorem:MTL expressiveness characterization}.(\ref{enum:main theorem:bounded bisimulation invariant})}

\begin{proof}
  The direction \ref{enum:main theorem:MTL expressible} $\rightarrow$ \ref{enum:main theorem:bounded bisimulation invariant} follows immediately from Proposition~\ref{proposition:team bisimulation implies equivalence}. For the converse, assume that $P$ is invariant under $k$-bisimulation. Without loss of generality assume $P\neq\emptyset$. We claim that the formula
  
  $$\varphi_P:=\ovee_{(M,T)\in P}\hintikkaModelWithTeam MTk$$ expresses $P$.
  
  First note that $\varphi_P$ can be written as  the disjunction of only finitely many formulas: Each $\hintikkaModelWithTeam MTk$ is uniquely defined by a subset of the finite set $\hintikkaFormulasLevel k$, therefore there are only finitely many formulas of the form $\hintikkaModelWithTeam MTk$. 
  
  We now show that for each model $M$ and team $T$, we have that $(M,T)\in P$ if and only if $M,T\models\varphi_P$. First assume that $(M,T)\in P$. Then the fact that $(M,w)\bisim k(M,w)$ for each model $M$, each world $w$ and each number $k$ and Proposition~\ref{proposition:team bisimulation characterization with hintikka formulas} imply that $M,T\models\hintikkaModelWithTeam MTk$. Therefore, $M,T\models\varphi_P$. For the converse, assume that $M,T\models\varphi_P$. Then there is some $(M',T')\in P$ with $M,T\models\hintikkaModelWithTeam{M'}{T'}k$. Due to Proposition~\ref{proposition:team bisimulation characterization with hintikka formulas}, it follows that $(M,T)\bisim k(M',T')$. Since $P$ is invariant under $k$-bisimulation, it follows that $(M,T)\in P$ as required.
\end{proof}

\subsubsection{Proof of implication \ref{theorem:MTL expressiveness characterization}.(\ref{enum:main theorem:bounded bisimulation invariant}) $\rightarrow$ \ref{theorem:MTL expressiveness characterization}.(\ref{enum:main theorem:FO expressible and bisimulation invariant})}

\begin{proof}
 It suffices to show that $P$ can be expressed in first-order logic. This follows using essentially the standard translation from modal into first-order logic. Since classical disjunction is of course available in first-order logic, the proof of the implication \ref{enum:main theorem:bounded bisimulation invariant} $\rightarrow$ \ref{enum:main theorem:MTL expressible} shows that it suffices to express each $\hintikkaModelWithTeam MTk$ (expressing team-bisimilarity to $M,T$) in first-order logic. 
 
 Each of the Hintikka formulas $\hintikkaPointedModel Mwk$ (expressing bisimilarity to the pointed model $M,w$) is a standard modal formula, therefore an application of the standard translation gives a first-order formula $\phi_{M,w}^{k,\textrm{FO}}$
 with a free variable $x$ such that for all models $M'$ and worlds $w'$, we have that $M',w'\models\hintikkaPointedModel Mwk$ if and only if $\foStructure{M'}{\emptyset}\models \phi_{M,w}^{k,\textrm{FO}}(w)$. We now show how to express $\hintikkaModelWithTeam MTk$ (expressing team-bisimilarity to $M,T$) in first-order logic.
 
 Recall that $\hintikkaModelWithTeam MTk$ is defined as $\left(\bigwedge_{\varphi\in\hintikkaSetOfPointedModelsInTeam MTk}\existsWorld\varphi\right)\wedge\left(\bigvee_{\varphi\in\hintikkaSetOfPointedModelsInTeam MTk}\varphi\right)$. Therefore, a first-order representation of $\hintikkaModelWithTeam MTk$ is given as
 $$\left(\bigwedge_{\varphi\in\hintikkaSetOfPointedModelsInTeam MTk}\exists w (T(w)\wedge\varphi^{\mathtext{FO}}(w))\right)\wedge\left(\forall w (T(w)\implies\bigvee_{\varphi\in\hintikkaSetOfPointedModelsInTeam MTk}\varphi^{\mathtext{FO}}(w)\right),$$
 where $\varphi^{\mathtext{FO}}$ is the standard translation of $\varphi$ into first-order logic as mentioned above. This concludes the proof.
\end{proof}

\subsubsection{Proof of implication \ref{theorem:MTL expressiveness characterization}.(\ref{enum:main theorem:FO expressible and bisimulation invariant}) $\rightarrow$ \ref{theorem:MTL expressiveness characterization}.(\ref{enum:main theorem:bisimulation invariant and local})}

\begin{proof}
 Let $\varphi$ be the first-order formula expressing $P$. Since $\varphi$ is first-order, we know that $\varphi$ is Hanf-local\footnote{See Appendix~\ref{appendix:hanf locality} for a brief discussion of Hanf-locality}. Let $d$ be the Hanf-locality rank of $\varphi$. We show that $\varphi$ is $2d$-local. Therefore, let $M$ be a model with team $T$. We show that $\foStructure MT\models\varphi$ if and only if $\foStructure{\neighborhood T{2d}M}{T}\models\varphi$. Since $\varphi$ is bisimulation-invariant, it suffices to construct models $M_1$ and $M_2$ containing $T$ such that
 \begin{itemize}
  \item $(M_1,T)$ and $(M,T)$ are team-bisimilar,
  \item $(M_2,T)$ and $(\neighborhood T{2d}M,T)$ are team-bisimilar,
  \item $\foStructure{M_1}{T}\models\varphi$ if and only if $\foStructure{M_2}{T}\models\varphi$.
 \end{itemize}

 We first define $\mDiss$ as the model obtained from $M$ by disconnecting $\neighborhood T{2d}M$ from the remainder of the model, i.e., by removing all edges between $\neighborhood T{2d}M$ and $M\setminus\neighborhood T{2d}M$. Since $\mDiss$ is also obtained from $\neighborhood T{2d}M$ by adding the remainder of the model $M$ without connecting the added worlds to $\neighborhood T{2d}M$, it is obvious that $(\mDiss,T)\bisim{}(\neighborhood T{2d}M,T)$. We now define the models $M_1$ and $M_2$ such that $(M_1,T)\bisim{}(M,T)$ and $(M_2,T)\bisim{}(\mDiss,T)$ (and hence $(M_2,T)\bisim{}(\neighborhood T{2d}M,T)$) as follows:
 
 \begin{itemize}
  \item $M_1$ and $M_2$ are obtained from $M$ and $\mDiss$ by adding the exact same components: For each $w\in M$ (note that $M$ and $\mDiss$ have the exact same set of worlds), countably infinitely many copies of $\neighborhood w{2d}M$ and of $\neighborhood w{2d}{\mDiss}$ are added to both $M_1$ and $M_2$.
  \item for $n\in\mathbb N$, and $i\in\set{1,2}$, with $C^{\mathtext{DISS}}_{i,n}(w)$, we denote the $n$-th copy of $\neighborhood w{2d}{\mDiss}$ in $M_i$, the \emph{center} of $C^{\mathtext{DISS}}_{i,n}(w)$ is the copy of $w$ in $C^{\mathtext{DISS}}_{i,n}(w)$.
  \item for $n\in\mathbb N$, and $i\in\set{1,2}$, with $C^{\mathtext{CONN}}_{i,n}(w)$, we denote the $n$-th copy of $\neighborhood w{2d}{M}$ in $M_i$, the \emph{center} of $C^{\mathtext{CONN}}_{i,n}(w)$ is the copy of $w$ in $C^{\mathtext{CONN}}_{i,n}(w)$.
 \end{itemize}
 
 In the above, when we ``copy'' a part of a (Kripke) model, this includes copying the values of the involved propositional variables in these worlds (this is reflected in the resulting first-order models in the obvious way). However, we stress that the team $T$ is treated differently: The set $T$ is not enlarged with the copy operation, i.e., a copy of a world in $T$ is itself not an element of $T$.
 
 Since $M_1$ and $M_2$ are obtained from $M$ and $\mDiss$ by adding new components that are not connected to the original models, it clearly follows that $(M,T)$ and $(M_1,T)$ are team-bisimilar, and $(\mDiss,T)$ and $(M_2,T)$ are team-bisimilar. Note that each $w$ in the $M$-part of $M_1$ is the center of a $2d$-environment isomorphic to $C^{\mathtext{CONN}}_{2,n}(w)$, and each $w$ in the $\mDiss$-part of $M_2$ is the center of a $2d$-environment isomorphis to $C^{\mathtext{DISS}}_{1,n}(w)$. 
 
 Since the models $M$ ($\mDiss$) contain one copy of each $\neighborhood w{2d}M$ ($\neighborhood w{2d}{\mDiss}$), both $M_1$ and $M_2$ contain countably infinitely many copies of each $\neighborhood w{2d}M$ and each $\neighborhood w{2d}{\mDiss}$. Let $S_1$ be the subset of $M_1$ containing only the points from the $M$-part of $M_1$, plus the center of each $C^{\mathtext{CONN}}_{1,n}(w)$, and the center of each $C^{\mathtext{DISS}}_{1,n}(w)$. Similarly, let $S_2$ be the subset of $M_2$ containing only the points from the $\mDiss$-part of $M_2$ plus the centers of the added components. 
 
 Since $M_1$ and $M_2$ contain the same number of copies of each relevant neighborhood, there is a bijection $f\colon S_1\rightarrow S_2$ such that for each $w\in S_1$, the $2d$-neighborhoods of $w$ and $f(w)$ are isomorphic. Now $f$ can be modified such that for each $w\in M$ which has distance at most $d$ to a world in $T$, the value $f(w)$ is the corresponding world in the $\mDiss$-part of $M_2$. The thus-modified $f$ now satisfies that for each $w\in S_1$, the $d$-neighborhoods of $w$ and $f(w)$ are isomorphic. We can easily extend $f$ to worlds in $C^{\mathtext{DISS}}_{1,n}$ and $C^{\mathtext{CONN}}_{1,n}$ that are not the center of their respective components by mapping such a world $w$ in $C^{\mathtext{DISS}}_{1,n}$ to the copy of $w$ in $C^{\mathtext{DISS}}_{2,n}$, and analogously for $C^{\mathtext{CONN}}_{1,n}$.
 
 Therefore, we have constructed a bijection $f\colon M_1\rightarrow M_2$ such that for each $w\in M_1$, the $d$-neighborhood of $w$ in $M_1$ is isomorphic to the $d$-neighborhood of $w$ in $M_2$. Since $\varphi$ is Hanf-local with rank $d$, this implies that $\foStructure{M_1}{T}\models\varphi$ if and only if $\foStructure{M_2}{T}\models\varphi$, as required.
\end{proof}

The proof of this implication uses ideas from Otto's proof of van Benthem's classical theorem presented in~\cite{OTTO-ELEMENTARY-VAN-BENTHEM-PROOF-TR-2004}. However our proof is based on the Hanf-locality of first-order expressible properties, whereas Otto's proof uses Ehrenfeucht-Fra\"{\i}ss\'e games, as a consequence, our construction requires an infinite number of copies of each model due to cardinality reasons.

\subsubsection{Proof of implication \ref{theorem:MTL expressiveness characterization}.(\ref{enum:main theorem:bisimulation invariant and local}) $\rightarrow$ \ref{theorem:MTL expressiveness characterization}.(\ref{enum:main theorem:bounded bisimulation invariant})}

\begin{proof}
 Assume that $P$ is invariant under bisimulation, and $P$ is $k$-local for some $k\in\mathbb N$. We show that $P$ is invariant under $k$-bisimulation. Hence let $M_1,T_1\bisim kM_2,T_2$. Since $P$ is invariant under bisimulation, we can without loss of generality assume that $M_1$ and $M_2$ are directed forests, that $M_1$ contains only worlds connected to worlds in $T_1$, and analogously for $M_2$ and $T_2$. Since $P$ is also $k$-local, we can also assume that $M_1$ contains no world with a distance of more than $k$ to $T_1$, and analogously for $M_2$ and $T_2$. From these assumptions, it immediately follows that $M_1,T_1\bisim{}M_2,T_2$, and, since $P$ is invariant under bisimulation, this implies that $(M_1,T_1)\in P$ if and only if $(M_2,T_2)$ in $P$, as required.
\end{proof}

\section{Alternative logical characterisations for the bisimulation invariant properties}

Research on variants of (modal)
dependence logic has concentrated on logics defined in terms of independence and inclusion atoms. Analogously to \MDL, these logics are invariant under bisimulation but are strictly less expressive than \MTL \cite{DBLP:conf/aiml/KontinenMSV14}. On the other hand, extended modal dependence logic, \EMDL,  uses  dependence atoms but allows them to be  applied to \ML-formulas instead of just proposition symbols \cite{DBLP:conf/wollic/EbbingHMMVV13}. This variant is also known to be a proper sub-logic of \MTL being able to express  all downwards-closed properties that are invariant under $k$-bisimulation for some $k\in \mathbb{N}$, and equivalent to $\ML(\ovee)$ \cite{DBLP:conf/aiml/HellaLSV14}.

In this section we systematically study the expressive power of variants of  \EMDL replacing dependence atoms by independence and inclusion atoms. Depending on whether we also allow classical disjunction or not, this gives four logics, namely 
 $\EMIL$ (Extended Modal Independence Logic), $\EMILC$ (\EMIL extended with classical disjunction), $\EMINCL$ (Extended Modal Inclusion Logic) and $\EMINCLC$ (\EMINCL extended with classical disjunction).  We study the expressiveness of these logics, and show that while $\EMINCLC$ is as expressive as \MTL, for each of the other three logics there is an \MTL-expressible property that cannot be expressed in the logic. In the last section,  we also study the extension of \ML by   first-order definable generalised dependence  atoms, and show that the resulting logic---even without the addition of classical disjunction---is equivalent to  \MTL. 

\subsection{Extended Modal Independence Logic (\EMIL)}\label{sect:emil}

We first consider \emph{Extended Modal Independence Logic} (\EMIL). Syntactically, \EMIL extends \ML by the following: If $\overline P$, $\overline Q$, and $\overline R$ are finite sets of \ML-formulas, then $\overline P \bot_{\overline R}\overline Q$ is an \EMIL-formula. The semantics of this \emph{extended independence atom} are defined by lifting the definition for propositional variables given in~\cite{DBLP:conf/aiml/KontinenMSV14} to \ML-formulas as follows.

For a formula $\varphi$ and a world $w$, we write $\varphi(w)$ for the function defined as $\varphi(w)=1$ if $M,\set{w}\models\varphi$, and $\varphi(w)=0$ otherwise (the model $M$ will always be clear from the context). For a set of formulas $\overline F$ and worlds $w_1,w_2$, we write $w_1\equiv_{\overline F}w_2$ if $\varphi(w_1)=\varphi(w_2)$ for each $\varphi\in\overline F$.

$$
\begin{array}{lll}
M, T \models \overline P\bot_{\overline R}\ \overline{Q} & \Leftrightarrow & \forall w,w' \in T \colon w \equiv_{\overline R} w' \text { implies } \exists w'' \in T \colon \\
			&& w'' \equiv_{\overline{P}} w \text{ and } w'' \equiv_{\overline{Q}} w' \text{ and } w'' \equiv_{\overline{R}} w.
\end{array}		
		$$
The  extension of \EMIL by classical disjunction 	$\ovee$ is denoted by  $\EMILC$. 	

We will next show that  $\EMILC$ is a proper sub-logic of \MTL. The following lemma will be used in the proof.

\begin{lemma} \label{emptyR}
Let  $M=(W,R,\pi)$ be a Kripke model such that $R=\emptyset$ and $T\subseteq W$ a team. Then for all $\varphi \in \EMILC$ it holds that if $M,T\models \varphi$, then $M,\{w\}\models \varphi$ for all $w\in T$.
\begin{proof} A straight-forward induction on the construction of $\varphi$ using the facts that a singleton team trivially satisfies all independence atoms, and the empty team satisfies all formulas of  $\EMILC$.
\end{proof}
\end{lemma}

\begin{restatable}{theorem}{theoremnotexpressibleinemil}\label{theorem:not expressible in emil} 
 $\EMDL\subsetneq \EMIL\subseteq  \EMILC \subsetneq \MTL$.
\end{restatable}

\begin{proof} The first inclusion follows from the fact that dependence atoms can be expressed by independence atoms. The inclusion is strict since  \EMDL is downwards-closed and \EMIL is not.

For the last inclusion, note that every property expressible in $\EMILC$ is invariant under bisimulation, hence  it follows that \MTL can express every $\EMILC$-expressible property due to Theorem~\ref{theorem:MTL expressiveness characterization}. For the strictness, we show that there is a team property that is invariant under $0$-bisimulation and which cannot be expressed in $\EMILC$. Let $P$ be the property
$$P:=\set{(M,T)\ \vert\ M,T\models\existsWorld p},$$
 i.e., the class of $(M,T)$  such that $T$ contains at least one world in which  $p$ is satisfied. 
 Consider the model $M$ with worlds $w_1$ and $w_2$, where $p$ is true in $w_1$ and false in $w_2$, and the accessability relation $R=\emptyset$. Let $T_1=\set{w_1,w_2}$, and let $T_2=\set{w_2}$. Obviously, $(M,T_1)\in P$ and $(M,T_2)\notin P$. By Lemma \ref{emptyR} for all $\EMILC$-formulas $\varphi$, if  $M,T_1\models\varphi$, then $M,T_2\models\varphi$. This shows  that there is no $\EMILC$-formula expressing $P$.
\end{proof}

\subsection{Extended Modal Inclusion Logic}\label{sect:emincl}

Analogously to \EMIL, we now define \emph{Extended Modal Inclusion Logic}, \EMINCL. \EMINCL extends the syntax of \ML with the following rule: If $\varphi_1,\dots,\varphi_n$ and $\psi_1,\dots,\psi_n$ are \ML-formulas, then $(\varphi_1,\dots,\varphi_n)\subseteq(\psi_1,\dots,\psi_n)$ is an \EMINCL-formula. The semantics of this \emph{inclusion atom} are lifted from the first-order setting \cite{DBLP:journals/apal/Galliani12} to the extended modal case:

\begin{quote}
$M,T\models(\varphi_1,\dots,\varphi_n)\subseteq(\psi_1,\dots,\psi_n)$ if for every world $w\in T$ there is a world $w'\in T$ such that $\varphi_i(w)=\psi_i(w')$ for each $i\in\set{1,\dots,n}$.
\end{quote}
The extension of \EMINCL by classical disjunction 	$\ovee$ is denoted by  $\EMINCLC$.

Analogously to first-order inclusion logic  \cite{DBLP:conf/csl/GallianiH13},
the truth of \EMINCL-formulas is preserved under unions of teams. Hence we get the following  result. 

\begin{restatable}{theorem}{EMINCLlessthanMTL}
 \EMINCL is strictly less expressive than \MTL.
\end{restatable}

\begin{proof}
 We show that there is a $0$-bisimulation invariant property that cannot be expressed with \EMINCL. For this, let $P$ be the property 
 $$\set{(M,T)\ \vert\ \mathtext{there is exactly one }i\in\set{1,2}\mathtext{ with }M,T\models\existsWorld p_i}.$$ 
 Clearly (and also due to Theorem~\ref{theorem:MTL expressiveness characterization}), $P$ is invariant under $0$-bisimulation. Now, let $M$ be a model with worlds $w_1$ and $w_2$ such that in $w_i$, the variable $p_i$ is true and $p_{3-i}$ is false. Let $T_1=\set{w_1}$, and $T_2=\set{w_2}$. Then, by construction, $(M,T_1)\in P$ and $(M,T_2)\in P$, but $(M,T_1\cup T_2)\notin P$. Now assume that $\varphi$ is an \EMINCL-formula that expresses $P$.
 
 Then, in particular $M,T_1\models\varphi$, $M,T_2\models\varphi$ and $M,(T_1\cup T_2)\nmodels\varphi$. However, it easily follows that \EMINCL is \emph{union}-closed, i.e. if $M,T_1\models\varphi$ and $M,T_2\models\varphi$, then also $M,(T_1\cup T_2)\models\varphi$ (see, e.g.,~\cite{DBLP:conf/csl/GallianiH13}, the property trivially transfers to the modal setting). Therefore, we have a contradiction. 
\end{proof}

Next we want to show that \EMINCLC is as powerful as \MTL.

\begin{theorem}\label{theorem:EMINCLC expressiveness}
 Let $P$ be a team property. Then the following are equivalent:
 \begin{enumerate}
  \item\label{enum:theorem:EMINCLC expressiveness:bisimulation invariance} $P$ is invariant under $k$-bisimulation.
  \item\label{enum:theorem:EMINCLC expressiveness:EMINCLC expressiveness} There is an \EMINCLC-formula $\varphi$ with $\md\varphi=k$ that characterizes $P$.
 \end{enumerate}
\end{theorem}

\begin{proof}
 The direction from~\ref{enum:theorem:EMINCLC expressiveness:EMINCLC expressiveness} to~\ref{enum:theorem:EMINCLC expressiveness:bisimulation invariance} follows by a straight-forward extension of the proof of Proposition~\ref{proposition:team bisimulation implies equivalence}. For the converse, assume that $P$ is invariant under $k$-bisimulation. From the proof of Theorem~\ref{theorem:MTL expressiveness characterization}, we know that it suffices to construct an \EMINCLC-formula $\varphi$ that is equivalent to the \MTL-formula $\ovee_{(M,T)\in P}\hintikkaModelWithTeam MTk$. Since the $\ovee$-operator is available in $\EMINCLC$, it suffices to show how to express the formula $\hintikkaModelWithTeam MTk$ for each model $M$ and team $T$ as an \EMINCLC-formula. Recall that 
 
 $$\hintikkaModelWithTeam MTk=\left(\bigwedge_{\varphi\in\hintikkaSetOfPointedModelsInTeam MTk}\existsWorld\varphi\right)\wedge\left(\bigvee_{\varphi\in\hintikkaSetOfPointedModelsInTeam MTk}\varphi\right).$$
 
 The second conjunct already is an \EMINCLC-formula, hence it suffices to show how $\existsWorld\varphi$ can be expressed for an \ML-formula $\varphi$. As discussed earlier, $M,T\models\existsWorld\varphi$ for an \ML-formula $\varphi$ if and only if there is a world $w\in T$ with $M,\set w\models\varphi$. Hence from the semantics of the inclusion atom, it is clear that $\existsWorld\varphi$ is equivalent to $(x\vee\neg x)\subseteq(\varphi)$. This concludes the proof.
\end{proof}

\subsection{\ML with  FO-definable generalized dependence atoms}

In this section we show that \MTL, and the bisimulation invariant properties, can be captured as the extension of \ML by all  generalized dependence atoms definable in first-order logic without identity.
The notion of a generalized dependence atom in the modal context was introduced in  \cite{DBLP:conf/aiml/KontinenMSV14}. A  closely related notion was introduced and studied in the first-order context in \cite{ku13}. The semantics of a generalized dependence atom $D$ is determined essentially by a property  of teams. 

In the following we are interested in generalized dependence atoms definable by first-order formulae, defined as follows:  Suppose that  $D$ is  an atom of width $n$, that is, an atom that applies to $n$ propositional variables (for example the atom in \eqref{modaldepatom}). We say that $D$ is \FO-definable if there exists a \FO-sentence $\phi$ over signature $\langle A_{1}, \dots, A_{n}\rangle$ such that for all Kripke models $M= (W,R,\pi)$ and teams $T$,
		$$M, T \models D(p_1, \dots, p_n) \;\Longleftrightarrow\; \mathcal{A} \models \phi,$$
	where $\mathcal{A}$ is the first-order structure  with universe $T$ and relations $A_{i}^{\mathcal{A}}$ for $1 \leq i \leq n$, where for all $w\in T$, $w \in A_{i}^{\mathcal{A}} \Leftrightarrow p_i \in \pi(w)$.

In our ``extended'' setting  the arguments to  a generalized dependence atom $D(\varphi_1,\dots,\varphi_n)$ can be arbitrary \ML-formulas instead of propositional variables. Hence the relation $A_i$ is now interpreted by the worlds of $T$ in which $\varphi_i$ is satisfied. We denote by  \MLFO the extension of  \ML by all generalized dependence atoms $D$ that are \FO-definable without identity. 

\begin{theorem}\label{thm:MTL=ML^FO}
 \MLFO is equally expressive as \MTL.
\end{theorem}

\begin{proof}
In the proof of Theorem~6.8 in~\cite{DBLP:conf/aiml/KontinenMSV14} it is showed that \MLFO is invariant under bisimulation in the case where  generalised atoms may be applied only to propositional variables.  The proof easily extends to the setting where arbitrary \ML-formulas may appear as arguments to a generalised dependence atom. Therefore,  \MLFO is not more expressive than \MTL. For the converse, let $P$ be a property that can be expressed in \MTL. From Theorem~\ref{theorem:MTL expressiveness characterization}, it follows that $P$ is invariant under $k$-bisimulation, and from the proof of Theorem~\ref{theorem:MTL expressiveness characterization}, we know that it suffices to express the formula $\ovee_{(M,T)\in P}\hintikkaModelWithTeam MTk$ in \MLFO. We can do this with the following first-order definable atom (by suitably choosing the parameters $n,m\in \mathbb{N}$):
 
 \begin{quote}
 $M,T\models D(\varphi^1_1,\dots,\varphi^1_n,\varphi^2_1,\dots,\varphi^2_n,\dots,\varphi^m_1,\dots,\varphi^m_n)$ if and only if there is some $k\in\set{1,\dots,m}$ such that each $w\in T$ satisfies some $\varphi^k_i$, and for each $j\in\set{1,\dots,n}$, there is some $w\in T$ that satisfies $\varphi^k_j$.
 \end{quote}
 
The atom $D$ can now  be FO-defined by replacing the exists/for all quantifiers on the indices with disjunctions/conjunctions:
 
  $$\displaystyle\bigvee_{k\in\set{1,\dots,m}}
 \left(
 \forall x\left( A^k_{1}(x)\vee\dots\vee A^k_{n}(x)\right)
 \wedge
 \bigwedge_{j\in\set{1,\dots,n}}(\exists x\ A^k_{j}(x))\right)$$
 Then, the atom $D$ applied to the formulas in $\hintikkaModelWithTeam MTk$ for all $(M,T)\in P$ gives a formula expressing $P$.
\end{proof}

\section{Conclusion}

Our results show that, with respect to expressive power, modal team logic is a natural upper bound for all the logics studied so far in the area of modal team semantics. Overall, an interesting picture of the characterization of the expressiveness of modal logics in terms of bisimulation emerges: Let us say that ``invariant under bounded bisimulation'' means invariant under $k$-bisimulation for some finite $k$. Then we have the following hierarchy of logics:
\begin{itemize}
 \item Due to van Benthem's theorem~\cite{BenthemBook}, \ML can exactly express all properties of pointed models that are FO-definable and invariant under bisimulation.
 \item Due to~\cite{DBLP:conf/aiml/HellaLSV14}, \ML with team semantics and extended with classical disjunction $\ovee$ can exactly express all properties of teams that are invariant under bounded bisimulation and additionally downwards-closed.
 \item Our result shows that \ML with team semantics and extended with classical negation $\classNegation$ can exactly express all properties of teams that are invariant under bounded bisimulation.
\end{itemize}

A number of open questions in the realm of modal logics with team semantics remain:
\begin{enumerate}
\item 
In the proof of Theorem~\ref{thm:MTL=ML^FO}, for each $k$, there is only a finite width of the $D$-operator above needed to express all properties that are invariant under $k$-bisimulation. However, the theorem leaves open the question whether  there is a ``natural'' atom $D$ or an atom with ``restricted width'' that gives the entire power of \MTL.\item 
Can we axiomatize \MTL? Axiomatizability of  sublogics  of \MTL has been studied, e.g.,  in  \cite{Yangthesis} and \cite{SanoV14}.
\item While we mentioned a number of complexity results on modal dependence logic and some of its extensions, this issue remains unsettled for full \MTL. In particular, what is the complexity of satisfiability and validity of \MTL?
\end{enumerate}

\begin{appendix}

\section{Hanf-locality of first-order logic}\label{appendix:hanf locality}

In the proof of the implication Proof of implication \ref{enum:main theorem:FO expressible and bisimulation invariant} $\rightarrow$ \ref{enum:main theorem:bisimulation invariant and local} of Theorem~\ref{theorem:MTL expressiveness characterization}, we used Hanf-locality of first-order formualas. In the following, we briefly introduce the relevant definitions and state the main result. Our presentation is based on~\cite[Section 4]{Libkin-ELEMENTS-FINITE-MODEL-THEORY-SPRINGER-2004}.

Since in this paper, we only consider properties of Kripke models, our first-oder models are defined over a relational signature that includes (at most) a single binary relation $R$ representing the Kripke relation of our models, a unary relation $T$ representing a team in a Kripke model, and for each propositional variable, a unary relation representing the worlds in which the variable is true. In the following, we assume a fixed finite set of variables.

For a world $w\in M$, and a natural number $r$, the \emph{radius $r$ ball} around $w$ is the set of worlds $w'$ such that there is an (undirected) path in $M$ between $w$ and $w'$ with length at most $r$. The $r$-neighborhood of $w$ in $M$, denoted with $N^M_r(w)$, is the model $M$ restricted to the radius $r$ ball around $w$, with an additional constant interpreted as $w$.

For two models $M_1$ and $M_2$, we write $M_1\leftrightarrows d M_2$, if there is a bijection $f\colon W_1\rightarrow W_2$ such that for every $c\in W_1$, $N^{M_1}_d(c)\cong N^{M_2}_d(f(c))$.

\begin{definition}
 A property $P$ of Kripke models is \emph{Hanf-local} if there is some $d\ge0$ such that for every Kripke models $M_1$ and $M_2$ with $M_1\leftrightarrows_d M_2$, we have that $M_1$ has property $P$ if and only if $M_2$ has property $P$. The smallest number $d$ for which this is true is the Hanf-locality rank of $P$.
\end{definition}

In our proof, we use the property that first-order logic is Hanf-local:

\begin{theorem}\cite{hanf65}.
 If $P$ is a property of Kripke models that can be expressed in first-oder logic, then $P$ is Hanf-local.
\end{theorem}

\end{appendix}

\end{document}